\documentclass[preprint,1p,11pt]{IR-Template/ISAS_IR}

\usepackage{bm} 
\usepackage{amsmath}
\usepackage{amsfonts}
\usepackage{amssymb}
\usepackage{dsfont} 
\usepackage{graphicx}
\usepackage{subfigure}
\usepackage[hidelinks]{hyperref} 
\usepackage[absolute,overlay]{textpos}
\usepackage{theorem}
\usepackage{color}
\usepackage{booktabs}
\usepackage{enumitem}

\usepackage{multirow}

\usepackage{bbm} 
\usepackage[protrusion=false]{microtype}

\usepackage[ruled,vlined]{algorithm2e}

\newcommand{\rv}[1]{\ensuremath{{\boldsymbol{#1}}}}
\newcommand{\nvec}[1]{\ensuremath{{{\underline{#1}}}}}
\newcommand{\rvec}[1]{\ensuremath{{\boldsymbol{\underline{#1}}}}}
\newcommand{\mat}[1]{{\ensuremath{{\mathbf{#1}}}}}

\newcommand{\pinv}{^\dagger}
\newcommand{\tr}{^{\top}}

\DeclareMathOperator{\ttrace}{trace}
\newcommand{\trace}[1]{\ttrace\bklammer{#1}}
\newcommand{\expect}[1]{\mathrm{E}\left\lbrace #1 \right\rbrace}

\newcommand{\IN}{\mathbb{N}}  
\newcommand{\IR}{\mathbb{R}}  

\DeclareMathOperator{\vvec}{vec}
\renewcommand{\vec}[1]{\vvec\klammer{#1}}
\DeclareMathOperator{\ddevec}{devec}
\newcommand{\devec}[1]{\ddevec\klammer{#1}}
\DeclareMathOperator{\st}{subject\ to}

\newcommand{\zeromatrix}{\mat{0}}
\newcommand{\identitymatrix}{\mat{I}}

\newcommand{\klammer}[1]{\left( #1 \right)}
\newcommand{\bklammer}[1]{\left[ #1 \right]}

\newcommand{\indicator}[1]{\mathbbm{1}_{#1}}
\newcommand{\kron}[2]{{#1 \otimes #2}}
\DeclareMathOperator{\diag}{diag}

	\newcommand{\Asys}[1]{\mat{A}_{#1}}
	\newcommand{\Bsys}[1]{\mat{B}_{#1}}
	\newcommand{\Qsys}[1]{\mat{Q}_{#1}}
	\newcommand{\Rsys}[1]{\mat{R}_{#1}}
	\newcommand{\Hsys}[1]{\mat{H}_{#1}}

	\newcommand{\xsys}[1]{\rvec{x}_{#1}}
	\newcommand{\usys}[1]{\nvec{u}_{#1}}
	
	\newcommand{\wsys}[1]{\rvec{w}_{#1}}

	\newcommand{\wCov}{\mat{W}}

	\newcommand{\Aaug}[1]{\widetilde{\mat{A}}_{#1}}

	\newcommand{\bigN}[1]{\nvec{\rho}_{#1}}

	\newcommand{\mode}[1]{\rv{\theta}_{#1}}
	\newcommand{\modeis}[1]{\theta_{#1}}
	\newcommand{\modeest}[2]{\widehat{\theta}^{(#1)}_{#2}}
	\newcommand{\modeinf}{\widehat{\theta}_{\infty}}
	
	\newcommand{\spectralradius}[1]{\rho\klammer{#1}}
	\newcommand{\spectralradiusopt}[1]{\widetilde{\rho}\klammer{#1}}
	
	\newcommand{\bigM}{\mat{M}}

	\newcommand{\TransitionMatrix}[1]{\mat{T}_{#1}}
	\newcommand{\transitionprob}[1]{p_{#1}}
	
	\newcommand{\Xsysest}[2]{\mat{X}^{(#1)}_{#2}}	
		
	\newcommand{\ControlLaw}{\mat{L}}

	\newcommand{\NumModes}{M}		
	\newcommand{\terT}{K}
	\newcommand{\CostFunc}{\mathcal{J}}
	
	\renewcommand{\P}[1]{\mat{P}_{#1}}
	\newcommand{\Pest}[2]{\mat{P}^{(#1)}_{#2}}
	\newcommand{\Lagrange}[2]{\mat{\Lambda}^{(#1)}_{#2}}
	\newcommand{\Hamiltonian}{\mathcal{H}}
	
	\newcommand{\vectorizedState}[1]{\nvec{\psi}_{#1}}

\newtheorem{theorem}{Theorem}
\newtheorem{remark}{Remark}
\newtheorem{definition}{Definition}

\newtheorem{proof}{Proof}

\begin{document}
    \begin{frontmatter}
        \title{Infinite-horizon Linear Optimal Control of\\Markov Jump Systems without Mode Observation\\via State Feedback}
        
        \author{Maxim~Dolgov}
        \ead{maxim.dolgov@kit.edu}
        
        \author{Uwe~D.~Hanebeck}
        \ead{uwe.hanebeck@ieee.org}
        
        \address{Intelligent Sensor-Actuator-Systems Laboratory (ISAS)\\
                 Institute for Anthropomatics and Robotics\\
                 Karlsruhe Institute of Technology (KIT), Germany\vspace{3mm}}
        
        \begin{abstract}
        In this paper, we consider stochastic optimal control of Markov Jump Linear Systems with state feedback but without observation of the jumping parameter. The proposed control law is assumed to be linear with constant gains that can be obtained from the necessary optimality conditions using an iterative algorithm. The proposed approach is demonstrated in a numerical example.
        \end{abstract}
    \end{frontmatter}
    
   \section{Introduction}
   	\label{sec:Introduction}
   	Since their introduction by Krasovskii and Lidskii in 1969 \cite{Krasovskii_1961_1,Krasovskii_1961_2,Krasovskii_1961_3}, Markov Jump Linear Systems (MJLS) have received a considerable amount of interest. This is due to their ability to capture systems whose dynamics are subject to abrupt changes that are not independently distributed. MJLS modeling approach is used, e.g., in networked control~\cite{Hespanha_2007,ACC13_Fischer}, economics~\cite{doVal_1999,Elliott_2007}, or control of systems with component failures~\cite{Vargas_2013}.\\

Most works that consider control of MJLS assume availability of the jumping parameter or \emph{mode} that models the abrupt model switching. This assumption allows to derive optimal control laws in continuous~\cite{Sworder_1969} and discrete time~\cite{Chizeck_1986,Fragoso_1989} for systems with state feedback. For measurement-feedback case, mode availability guarantees that the separation between control and estimation holds. Thus, the optimal control law consists of an optimal linear regulator and an optimal Kalman filter~\cite{Chizeck_1988}.\\

However, if the mode is not available, the control law becomes nonlinear because of the dual effect~\cite{Griffiths_1985,Casiello_1989}. In this case, the optimal solution is computationally intractable due to the curse of dimensionality. Thus, research concentrates on approximate control laws. We distinguish between two classes of approaches: (i) approaches based on assumed separation and (ii) approaches based on structural assumptions. Approaches that belong to the first class approximate the involved conditional densities. By doing so, it is possible to establish separation. Then, the optimal control law consists of an estimator and a regulator whose gains are linear. The estimator is either based on an Interacting Multiple Model (IMM) algorithm~\cite{Campo_1992} or on a Viterbi-like algorithm~\cite{Gupta_2003}. Approaches that belong to the second class make an assumption considering the control law, usually that the control law is linear such as in~\cite{doVal_1999,Costa_2000,Vargas_2013,Vargas_2014}.\\

Between full mode observation and no mode observation is the clustered mode observation. The term clustered can refer to (1) temporal interchange between full mode observation and no mode observation, and (2) observation of subsets of modes, i.e., observation whether one of the modes in a subset is active or not. We will not review this field in our paper. We refer an interested reader to, e.g.,~\cite{doVal_1999,doVal_2002} and the references therein.\\

In this paper, we take the approach (ii) and assume the controller to be linear and to possess constant regulator gain. Our approach differs from the works \cite{doVal_1999,Costa_2000,Vargas_2013,Vargas_2014} in the following way: \cite{doVal_1999,Costa_2000,Vargas_2013} assume time-variant controller gains and \cite{Vargas_2014} considers finite-horizon control with constant gains. And the works \cite{doVal_1999,Costa_2000,Vargas_2013,Vargas_2014} have to be implemented in a receding-horizon framework to be applicable for long operation times. The approach presented in \cite{Vargas_2014} can be used to compute constant gains. However, in this case, the optimization horizon becomes a parameter that must be chosen sufficiently large in order to obtain an infinite-horizon control law. To obtain the controller gain for the approach presented in this paper, we minimize an infinite-horizon cost function. By doing so, there is neither a need for choosing an optimization horizon, nor for implementing the control law in a receding-horizon framework. However, the latter can be done in order to, e.g., adapt the control law to changes in the system dynamics (both continuous- and discrete-valued), if desired. As we will see in the numerical example, the performance of the proposed controller, although it is time-invariant, can be almost identical to the performance of the receding-horizon time-variant controller from~\cite{Vargas_2013}.\\

The dynamics of the discrete-time MJLS considered in this paper are given by
\begin{align}
\begin{aligned}
\xsys{k+1}
	&=
	\Asys{\mode{k}}\xsys{k} + \Bsys{\mode{k}}\usys{k} + \Hsys{\mode{k}}\wsys{k}\ ,
\end{aligned}
\label{eq:SystemDynamics}
\end{align}
where $\xsys{k}\in\IR^n$ denotes the system state, $\usys{k}\in\IR^m$ the control input, and $\wsys{k}\in\IR^w$ the independent and identically distributed (i.i.d.) zero-mean second-order noise with covariance $\wCov = \expect{\wsys{k}\wsys{k}\tr}$. Here $\expect{\cdot}$ is the expectation operator and $\mat{A}\tr$ denotes the transpose of $\mat{A}$. The matrices $\Asys{\mode{k}}$, $\Bsys{\mode{k}}$, and $\Hsys{\mode{k}}$ are selected from time-invariant sets of matrices $\{\Asys{1},\cdots,\Asys{\NumModes}\}$, $\NumModes\in\IN$, etc. according to the jumping parameter $\mode{k}\in\{1,2,\cdots,\NumModes\}$ which is the state of a regular homogeneous Markov chain. We will refer to $\mode{k}$ as the \emph{mode}. The regularity assumption guarantees that the limit distribution $\modeinf$ of $\mode{k}$ exists~\cite{Grinstead_2003}.\\

The performance of the controlled system is measured by an infinite-horizon cost function
\begin{align}
\CostFunc
	&=
	\lim\limits_{\terT\rightarrow\infty} \frac{1}{\terT} \expect{\sum\limits_{k=0}^{\terT} \bklammer{\xsys{k}\tr\Qsys{\mode{k}}\xsys{k} + \usys{k}\tr\Rsys{\mode{k}}\usys{k}}}\ ,
\label{eq:CostFunc}
\end{align}
where for $i\in\{1,\dots,\NumModes\}$ the mode-dependent cost matrices $\Qsys{i}$ are positive semidefinite and $\Rsys{i}$ are positive definite, respectively.\\

The task is to find a control law that minimizes (\ref{eq:CostFunc}). As mentioned above, the optimal nonlinear control law that solves this task is computationally intractable. Thus, we make a structural assumption and choose the control law to be linear, mode-independent, and constant, i.e.,
\begin{align*}
\usys{k}
	&=
	\ControlLaw\xsys{k}\ .
\end{align*}
With this control law assumption, the considered problem can be formulated as
\begin{align}
\begin{aligned}
\min\limits_\ControlLaw&\hspace{2mm} \lim\limits_{\terT\rightarrow\infty} \frac{1}{\terT} \expect{\sum\limits_{k=0}^{\terT} \bklammer{\xsys{k}\tr(\Qsys{\mode{k}}+\ControlLaw\tr\Rsys{\mode{k}}\ControlLaw)\xsys{k}}} \\
\st&\hspace{2mm}
	\xsys{k+1} = \Asys{\mode{k}}\xsys{k} + \Bsys{\mode{k}}\usys{k} + \Hsys{\mode{k}}\wsys{k}\ ,\\
	&\hspace{2mm}\xsys{0}\ ,\ \modeis{0}\ ,\ \TransitionMatrix{}\ .
\end{aligned}
\label{eq:OptimizationProblem}
\end{align}

\emph{Outline.} The remainder of the paper is organized as follows. Before we present the main result in Sec.~\ref{sec:MainResult}, we introduce necessary definitions in the next section. A numerical example is given in Sec.~\ref{sec:NumericalExample} and Sec.~\ref{sec:Conclusion} concludes the paper.
%
%
%

   \section{Prerequisites}
   	\label{sec:Definitions}
Consider the MJLS 
\begin{align}
\xsys{k+1}
	&=
	\Asys{\mode{k}}\xsys{k}
\label{eq:StabSysDynamics}
\end{align}
with $\xsys{k}\in\IR^n$ being the system state, $\mode{k}\in\{1,\dots,\NumModes\}$ being the state of a regular, homogeneous Markov chain with transition matrix $\TransitionMatrix{}$, and $\Asys{\mode{k}}\in\{\Asys{1},\dots,\Asys{\NumModes}\}$.
\begin{definition}[Mean Square Stability]\hfill\\
System (\ref{eq:StabSysDynamics}) is mean square (MS) stable for any initial $\xsys{0}$ and $\mode{0}$, if it holds
\begin{align*}
\lim\limits_{k\rightarrow\infty} \left\|\expect{\xsys{k}\xsys{k}\tr}\right\| = 0\ .
\end{align*}
\label{def:MSS}
\end{definition}

\begin{remark}
If system~(\ref{eq:StabSysDynamics}) is affected by zero-mean second-order noise $\wsys{k}$ such that
\begin{align*}
\xsys{k+1}
	&=
	\Asys{\mode{k}}\xsys{k} + \Hsys{\mode{k}}\wsys{k}
\end{align*}
then the second moment $\expect{\xsys{k}\xsys{k}\tr}$ converges to a fixed point that is not $\zeromatrix$, i.e.,
\begin{align*}
\lim\limits_{k\rightarrow\infty} \expect{\xsys{k}\xsys{k}\tr} = \{\Qsys{1},\dots,\Qsys{\NumModes}\}\ ,
\end{align*}
where $\Qsys{i}$ are positive semidefinite. This claim can be shown using Banach's fixed point theorem (see~\ref{app:ProofConvergence}).
\end{remark}

The following theorem provides necessary and sufficient conditions for MS stability.
\begin{theorem}
For system~(\ref{eq:StabSysDynamics}), the two following conditions for mean square stability exist.
\begin{enumerate}[label = \alph*)]
\item 
	System~(\ref{eq:StabSysDynamics}) is MS stable, if for any positive matrices $\Qsys{1},\dots,\Qsys{\NumModes}$ there exist positive definite matrices $\P{1},\dots,\P{\NumModes}$ such that
	\begin{align*}
	\sum\limits_{i=1}^{\NumModes} \transitionprob{ij} \Asys{i}\tr\P{j}\Asys{i} - \P{i} = -\Qsys{i}\ ,
	\end{align*}
	where $\transitionprob{ij}$ denotes the transition probabilities from mode $i$ to mode $j$.
\item
	System~(\ref{eq:StabSysDynamics}) is MS stable, if for the spectral radius $\spectralradius{\cdot}$ of the matrix
	\begin{align}
		\bigM
			&=
			\klammer{\kron{\TransitionMatrix{}\tr}{\identitymatrix}} \diag\begin{bmatrix}(\kron{\Asys{1}}{\Asys{1}}) & \dots & (\kron{\Asys{\NumModes}}{\Asys{\NumModes}})\end{bmatrix}\ ,
	\label{eq:AugmentedMatrix}
	\end{align}
	where $\kron{}{}$ is the Kronecker product and $\diag$ the block diagonalization operator, it holds
	\begin{align*}
		\spectralradius{\bigM} < 1\ .
	\end{align*}
\end{enumerate}
\label{theor:MSS}
\end{theorem}
\begin{proof}
The proofs are given in~\cite{Morozan_1983,Fang_2002,Ling_2012} for systems with real-valued state and in~\cite{Costa_1993} for systems with complex-valued state.
\end{proof}

Next, we define MS stabilizability.
\begin{definition}[Mean Square Stabilizability]\hfill\\
System
\begin{align*}
\xsys{k+1}
	&=
	\Asys{\mode{k}}\xsys{k} + \Bsys{\mode{k}}\usys{k}\ ,
\end{align*}
with $\xsys{k}\in\IR^n$ being the system state, $\mode{k}\in\{1,\dots,\NumModes\}$ being the state of a regular, homogeneous Markov chain, and $\Asys{\mode{k}}\in\{\Asys{1},\dots,\Asys{\NumModes}\}$, $\Bsys{\mode{k}}\in\{\Bsys{1},\dots,\Bsys{\NumModes}\}$, is linearly mean square (MS) stabilizable without mode observation, if there exists a matrix $\ControlLaw$ such that
\begin{align*}
\xsys{k+1}
	&=
	\klammer{\Asys{\mode{k}}+\Bsys{\mode{k}}\ControlLaw}\xsys{k}
\end{align*}
is mean square stable.
\label{def:MSStab}
\end{definition}

   \section{Main Result}
   	\label{sec:MainResult}
Before we present the necessary optimality conditions for~(\ref{eq:OptimizationProblem}), we define the second-moment system state
\begin{align*}
\Xsysest{i}{k}
	&=
	\expect{\xsys{k}\xsys{k}\tr\indicator{\mode{k}=i}}\ ,
\end{align*}
where $\indicator{\mode{k}=i}=1$ if $\mode{k}=i$ and $0$ otherwise. The dynamics of the second-moment system state are
\begin{align}
\Xsysest{j}{k+1}
	&=
	\sum\limits_{i=1}^{\NumModes}\transitionprob{ij}\bklammer{(\Asys{i}+\Bsys{i}\ControlLaw)\Xsysest{i}{k}(\Asys{i}+\Bsys{i}\ControlLaw)\tr + \modeest{i}{k}\Hsys{i}\wCov\Hsys{i}\tr}\ ,
\label{eq:SecondMomentDynamics}
\end{align}
where $\modeest{i}{k}$ is the probability of being in mode $i$ at time step $k$.
\begin{theorem}[Necessary Optimality Conditions]\hfill\\
The necessary optimality conditions for the optimization problem~(\ref{eq:OptimizationProblem}) are given by
\begin{align}
	(\Asys{i}+\Bsys{i}\ControlLaw)\tr \Pest{i}{\infty} (\Asys{i}+\Bsys{i}\ControlLaw) + (\Qsys{i}+\ControlLaw\tr\Rsys{i}\ControlLaw) - \Lagrange{i}{\infty} &= \zeromatrix\ ,
\label{eq:CoupledEquationsP}\\
	\sum\limits_{j=1}^{\NumModes}\transitionprob{ji}\bklammer{(\Asys{j}+\Bsys{j}\ControlLaw)\Xsysest{j}{\infty}(\Asys{j}+\Bsys{j}\ControlLaw)\tr + \modeest{j}{k}\Hsys{j}\wCov\Hsys{j}\tr} - \Xsysest{i}{\infty} &= \zeromatrix\ ,
\label{eq:CoupledEquationsExx}\\
	\sum\limits_{i=1}^{\NumModes} \bklammer{(\Rsys{i}+\Bsys{i}\tr\Pest{i}{\infty}\Bsys{i})\ControlLaw\Xsysest{i}{\infty} + \Bsys{i}\tr\Pest{i}{\infty}\Asys{i}\Xsysest{i}{\infty}} &= \zeromatrix\ ,
\label{eq:CoupledEquationsL}
\end{align}
where $\Xsysest{i}{\infty} = \lim_{k\rightarrow\infty}\Xsysest{i}{k}$, $\Lagrange{\mode{k}}{\infty}\in\{\Lagrange{1}{\infty},\dots,\Lagrange{\NumModes}{\infty}\}$ are positive definite, and $\Pest{i}{\infty} = \sum_{j=1}^{\NumModes}\transitionprob{ij}\Lagrange{j}{\infty}$.
\label{theor:Main}
\end{theorem}
\begin{proof}
The proof is given in~\ref{app:ProofMainTheorem}.
\end{proof}

Please observe that equations~(\ref{eq:CoupledEquationsP}) constitute a set of coupled Riccati-like equations that reduce to the uncoupled Riccati equations if system~(\ref{eq:SystemDynamics}) has only one mode.\\

Finding a solution of~(\ref{eq:CoupledEquationsP})-(\ref{eq:CoupledEquationsL}) is not trivial. We propose to use a scheme similar to that presented in~\cite{DeKoning_1992} or~\cite{Bernstein_1987}. To this end, we first rewrite~(\ref{eq:CoupledEquationsL}) using the vectorization operator as
\begin{align*}
\klammer{\sum\limits_{i=1}^{\NumModes}\bklammer{\kron{\Xsysest{i}{\infty}}{(\Rsys{i}+\Bsys{i}\tr\Pest{i}{\infty}\Bsys{i}})}} \vec{\ControlLaw} + \vec{\sum\limits_{i=1}^{\NumModes} \Bsys{i}\tr\Pest{i}{\infty}\Asys{i}\Xsysest{i}{\infty}} = \nvec{0}\ .
\end{align*}
Solving for $\vec{\ControlLaw}$ yields
\begin{align*}
\vec{\ControlLaw}
	&=
	-\klammer{\sum\limits_{i=1}^{\NumModes}\bklammer{\kron{\Xsysest{i}{\infty}}{(\Rsys{i}+\Bsys{i}\tr\Pest{i}{\infty}\Bsys{i}})}}\pinv \vec{\sum\limits_{i=1}^{\NumModes} \Bsys{i}\tr\Pest{i}{\infty}\Asys{i}\Xsysest{i}{\infty}}\ ,
\end{align*}
where $\mat{A}\pinv$ denotes the Moore-Penrose pseudoinverse of $\mat{A}$.\\

The numerical algorithm is the following.\\

\begin{enumerate}[leftmargin = 1.4cm,rightmargin = 1cm,label= \hspace{1cm}\emph{Step \arabic*}:]
\item Initialize $\{\Xsysest{1}{[0]},\dots,\Xsysest{\NumModes}{[0]}\}$ and $\{\Lagrange{1}{[0]},\dots,\Lagrange{\NumModes}{[0]}\}$ with random values and compute $\allowbreak\{\Pest{1}{[0]}, \dots ,\Pest{\NumModes}{[0]}\}$.

\item Compute
\begin{align*}
\vec{\ControlLaw_{[k+1]}}
	&=
	-\klammer{\sum\limits_{i=1}^{\NumModes}\bklammer{\kron{\Xsysest{i}{[k]}}{(\Rsys{i}+\Bsys{i}\tr\Pest{i}{[k]}\Bsys{i}})}}\pinv \vec{\sum\limits_{i=1}^{\NumModes} \Bsys{i}\tr\Pest{i}{[k]}\Asys{i}\Xsysest{i}{[k]}}
\end{align*}
and reverse the vectorization operator in order to obtain $\ControlLaw_{[k+1]}$, i.e.,
\begin{align*}
\ControlLaw_{[k+1]}
	&=
	\devec{\vec{\ControlLaw_{[k+1]}}}
\end{align*}
with $\devec{\vec{\mat{A}}} = \mat{A}$.

\item Compute
\begin{align*}
\Xsysest{j}{[k+1]}
	&=
	\sum\limits_{i=1}^{\NumModes}\transitionprob{ij}\bklammer{(\Asys{i}+\Bsys{i}\ControlLaw_{[k+1]})\Xsysest{i}{[k]}(\Asys{i}+\Bsys{i}\ControlLaw_{[k+1]})\tr + \modeest{i}{\infty}\Hsys{i}\wCov\Hsys{i}\tr}\ ,\\
\Lagrange{i}{[k+1]}
	&=
	(\Asys{i}+\Bsys{i}\ControlLaw_{[k+1]})\tr \Pest{i}{[k]} (\Asys{i}+\Bsys{i}\ControlLaw_{[k+1]}) + (\Qsys{i}+\ControlLaw_{[k+1]}\tr\Rsys{i}\ControlLaw_{[k+1]})\ ,\\
\Pest{i}{[k+1]}
	&=
	\sum\limits_{i=1}^{\NumModes} \transitionprob{ij} \Lagrange{j}{[k+1]}\ .
\end{align*}
Stop if $\Xsysest{1}{[k]},\dots,\Xsysest{\NumModes}{[k]}$ converged. Otherwise, return to \emph{Step 2}.
\end{enumerate}

\begin{remark}
As in the case with i.i.d. system parameters considered in~\cite{DeKoning_1992}, convergence of the given algorithm does not always guarantee stability of the MJLS. Thus, it is always necessary to check if the computed control law stabilizes~(\ref{eq:SystemDynamics}) using Theorem~\ref{theor:MSS} with
\begin{align*}
\widetilde{\mat{A}}_{\mode{k}}
	&=
	\Asys{\mode{k}} + \Bsys{\mode{k}}\ControlLaw\ .
\end{align*}
\end{remark}

\begin{remark}
In order to check whether a MJLS is MS-stabilizable, it is possible to use the procedure described in Appendix~\ref{app:StabilizabilityTest}.
\end{remark}


   	
   \section{Numerical Example}
   	\label{sec:NumericalExample}
In order to demonstrate the performance of the proposed control approach, we performed Monte Carlo simulation runs with $100$ time steps each for different system and noise parameters. For each run, we computed the control law using different random initial guesses. The evolution of the mode and the noise were also randomly generated for each run. For comparison, we used the optimal controller published by Chizeck et al.~\cite{Chizeck_1986} that needs a mode feedback, and the finite-horizon controller without mode availability presented by Vargas et al.~\cite{Vargas_2013}.\\

The constant parameters of the simulated MJLS were chosen to
\begin{align*}
&\Asys{1} = \begin{bmatrix}1.2 & 1.2\\0& 1\end{bmatrix}\ ,\ \ 
\Asys{2} = \begin{bmatrix}1& 0.8\\0& 1\end{bmatrix}\ ,\ \ 
\Bsys{1} = \begin{bmatrix}0\\1\end{bmatrix}\ ,\ \ 
\Bsys{2} = \begin{bmatrix}0\\0.2\end{bmatrix}\ ,
\end{align*}
\begin{align*}
&\Hsys{1} = \identitymatrix\ ,\ \ 
\Hsys{2} = \identitymatrix\ ,\ \ 
\Qsys{1} = \identitymatrix\ ,\ \ 
\Qsys{2} = \identitymatrix\ ,\ \ 
\Rsys{1} = 1\ ,\ \ 
\Rsys{2} = 1\ .
\end{align*}
We considered two different noise scenarios with
\begin{align*}
\wCov_{1} = 0.01^2\ \ \text{and}\ \ \wCov_{2} = 0.5^2\ ,
\end{align*}
three different Markov chains with
\begin{align*}
\TransitionMatrix{1} = \begin{bmatrix}0.9 & 0.1\\ 0.1 & 0.9\end{bmatrix}\ ,\ \ 
\TransitionMatrix{2} = \begin{bmatrix}0.7 & 0.3\\ 0.6 & 0.4\end{bmatrix}\ ,\ \ 
\TransitionMatrix{3} = \begin{bmatrix}0.1 & 0.9\\ 0.3 & 0.7\end{bmatrix}\ ,
\end{align*}
and two different initial states
\begin{align*}
\xsys{0,1} = \begin{bmatrix}0 & 0\end{bmatrix}\tr\ \ \text{and}\ \ 
\xsys{0,2} = \begin{bmatrix}3 & 2\end{bmatrix}\tr\ .
\end{align*}
The spectral radii of the corresponding matrices constructed according to~(\ref{eq:AugmentedMatrix}) are 
\begin{align*}
\spectralradius{\mat{M}_1} = 1.3295\ ,\ \ 
\spectralradius{\mat{M}_2} = 1.2970\ ,\ \ \text{and}\ \ 
\spectralradius{\mat{M}_3} = 1.1047\ ,
\end{align*}
which shows that the MJLS is unstable for each of the transition matrices $\TransitionMatrix{1}$, $\TransitionMatrix{2}$, and $\TransitionMatrix{3}$.\\

Fig.~\ref{fig:ExampleRun} depicts the state trajectory, the applied control inputs, and the modes of the MJLS of an example run with $\TransitionMatrix{1}$, $\wCov_1$, and $\xsys{0} = \begin{bmatrix}3 & 2 \end{bmatrix}\tr$. Although the controller from~\cite{Vargas_2013} has time-variant gains while the proposed controller is time-invariant, the trajectories of both controllers are very similar.\\

The results of the Monte Carlo simulation with $\xsys{0} = \begin{bmatrix}0 & 0 \end{bmatrix}\tr$ are depicted in Fig.~\ref{fig:MC_x0_1} and the corresponding mean values of the costs are given in Table~\ref{tab:Costs_1}. In this scenario, the performance of the proposed controller and the controller from~\cite{Vargas_2013} is only slightly worse than the performance of the optimal controller with mode observation. And the performance of the proposed controller and the controller from~\cite{Vargas_2013} is almost equal. For $\xsys{0}=\begin{bmatrix}3&2\end{bmatrix}$, the simulation results are depicted in Fig.~\ref{fig:MC_x0_2} and the mean costs are given in Table~\ref{tab:Costs_2}. In this second scenario, the proposed control law performs well compared to the two other controllers if the noise covariance is large. However, the performance is worse if the noise covariance is low and the transition matrix is either $\TransitionMatrix{1}$ or $\TransitionMatrix{3}$. It is important to note that in contrast to the controller from~\cite{Vargas_2013}, the proposed controller is precomputed offline and does not depend on the initial state $\xsys{0}$ and the initial mode $\modeis{0}$. Thus, the computational footprint during operation is low.\\

\begin{figure}[h]
\centering
\includegraphics[width = .8\textwidth]{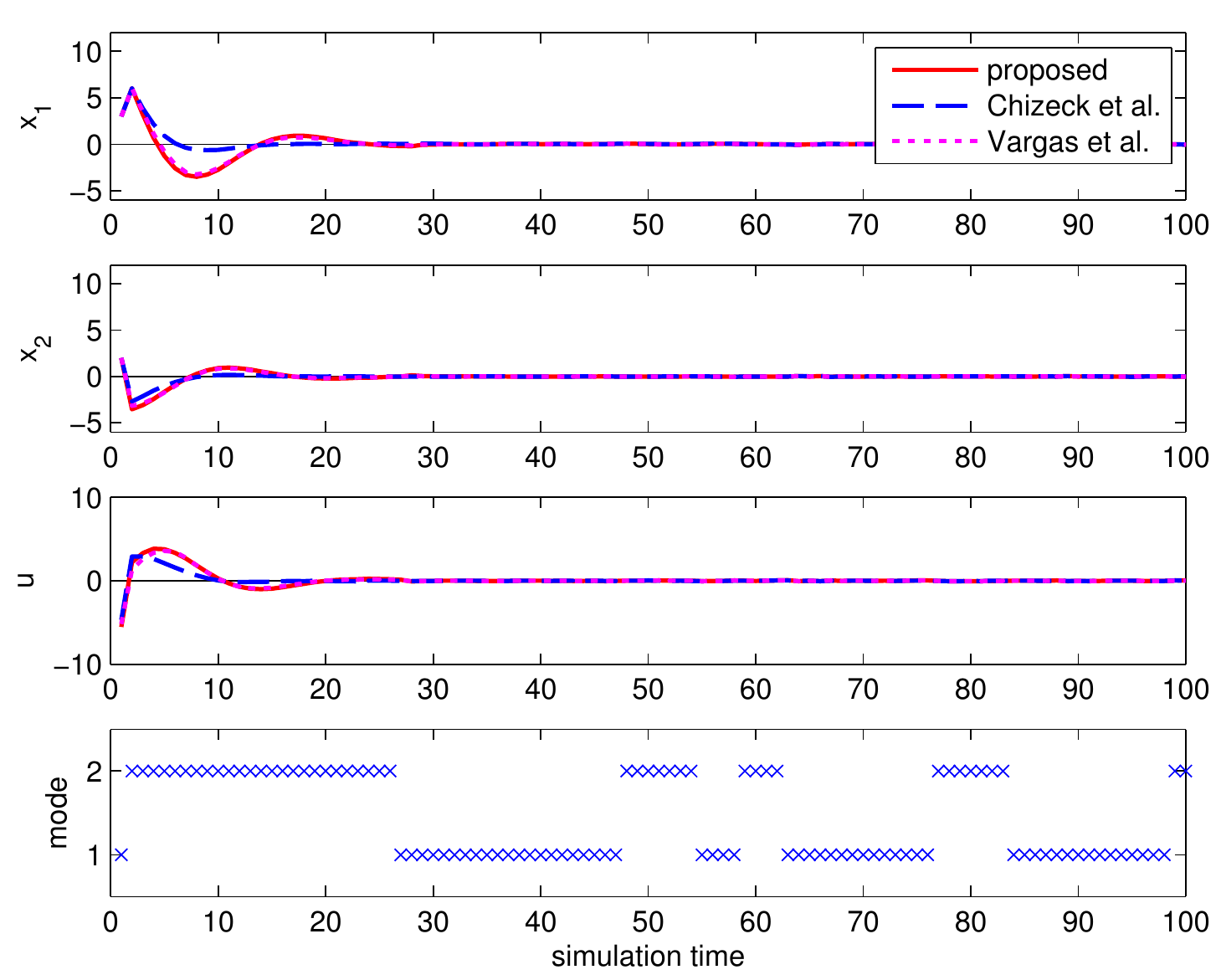}
\caption{Example run of the three compared controllers with $\TransitionMatrix{1}$, $\wCov_1$, and $\xsys{0} =\lbrack 3\hspace{2mm}  2\rbrack\tr$.}
\label{fig:ExampleRun}
\end{figure}

\begin{figure}[h]
\centering
\includegraphics[width = .9\textwidth]{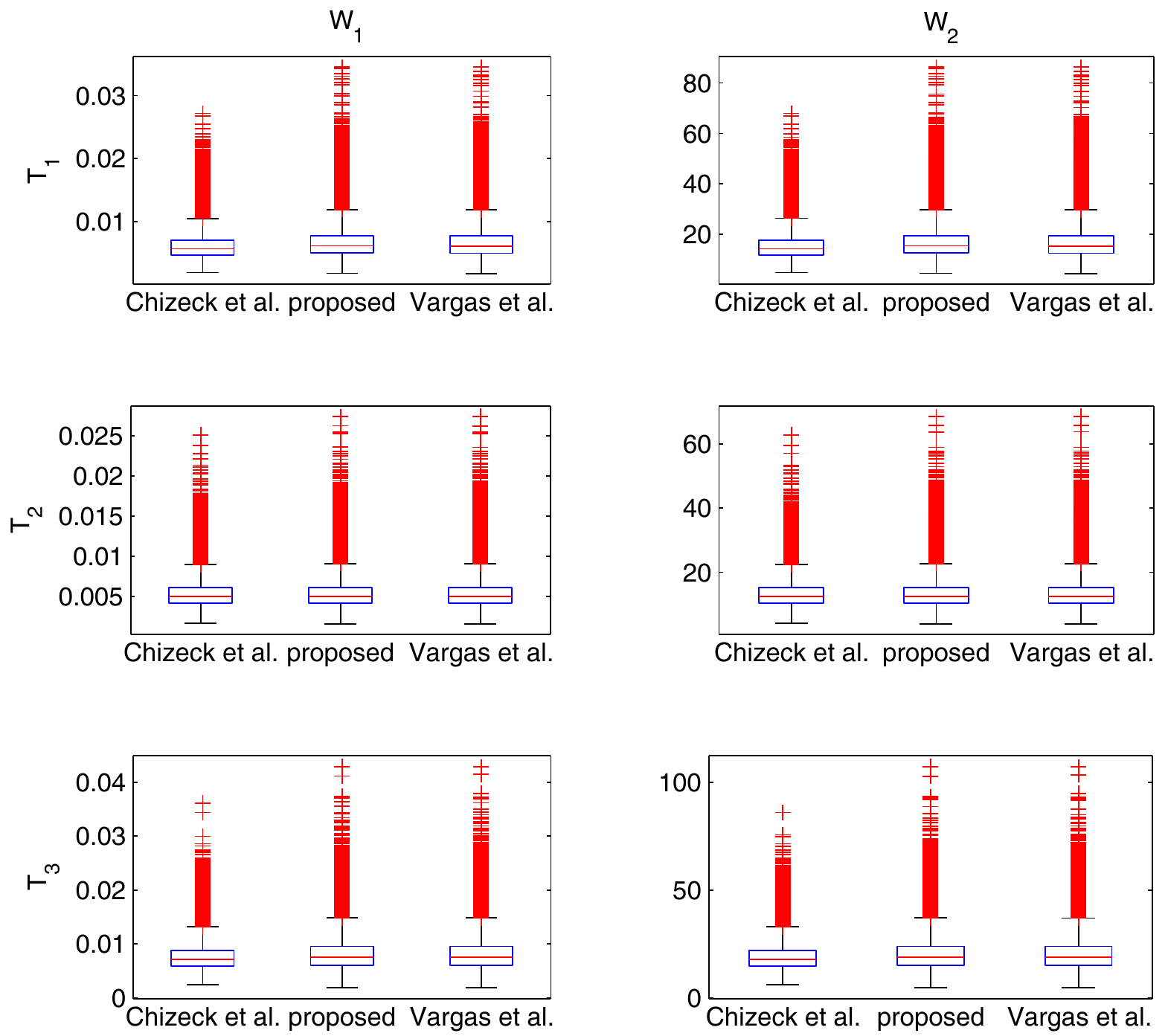}
\caption{Results of the Monte Carlo simulation. Depicted are the costs of the three compared controllers for different transition matrices and noise covariances, and $\xsys{0} =\lbrack 0\hspace{2mm}  0\rbrack\tr$.}
\label{fig:MC_x0_1}
\end{figure}

\begin{figure}[h]
\centering
\includegraphics[width = .9\textwidth]{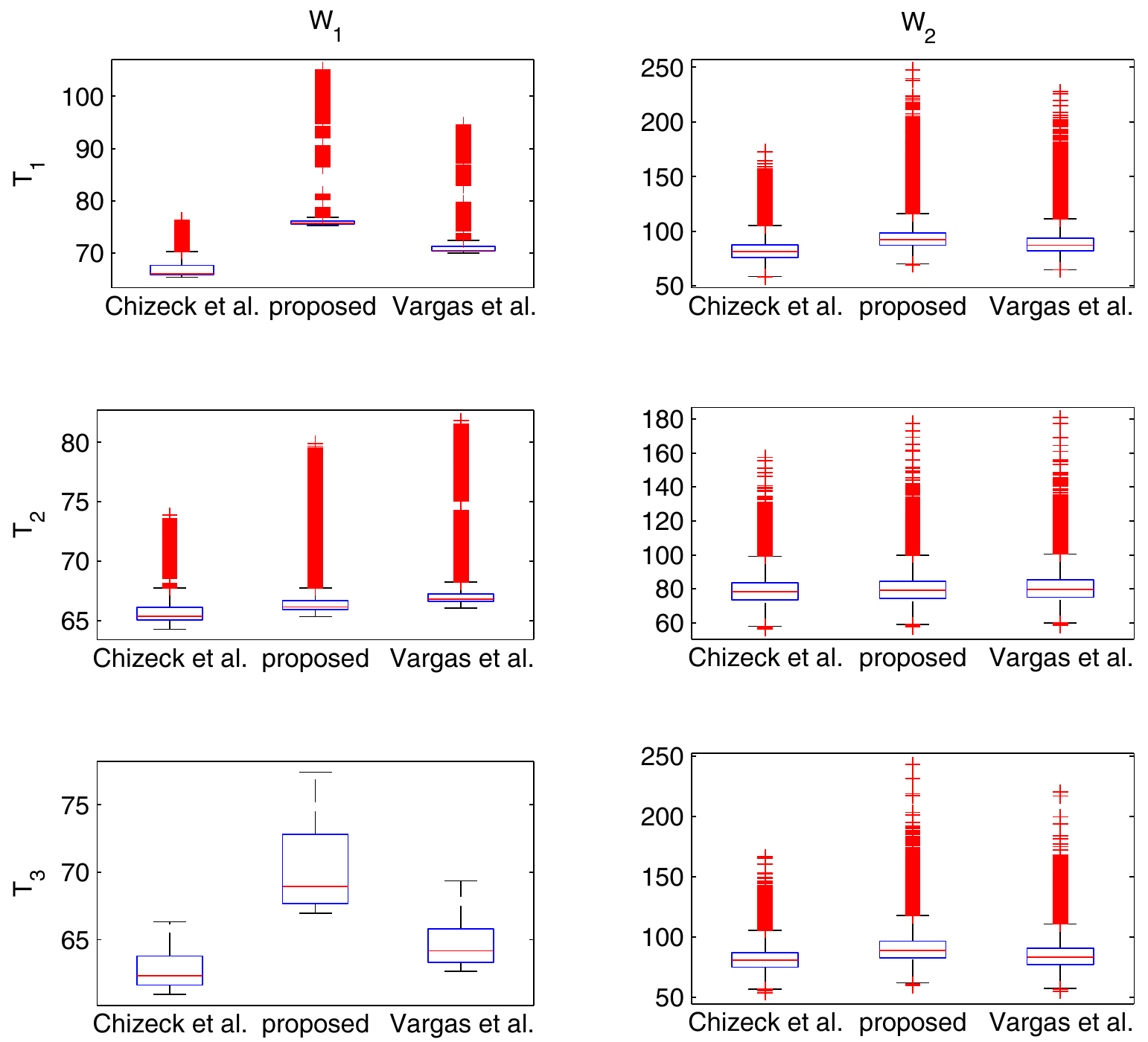}
\caption{Results of the Monte Carlo simulation. Depicted are the costs of the three compared controllers for different transition matrices and noise covariances, and $\xsys{0} =\lbrack 3\hspace{2mm}  2\rbrack\tr$.}
\label{fig:MC_x0_2}
\end{figure}

\begin{table}[h]
\centering
\begin{tabular}{ccccccc}
\toprule
& \multicolumn{3}{c}{$\wCov_1$} & \multicolumn{3}{c}{$\wCov_2$}\\ \cmidrule{2-7}
& Chizeck et al. & proposed & Vargas et al.& Chizeck et al. & proposed & Vargas et al.\\\cmidrule{1-7}
\multicolumn{1}{c}{$\TransitionMatrix{1}$}& $6.1005\text{e}^{-3}$ & $6.7175\text{e}^{-3}$ & $6.6756\text{e}^{-3}$ & $15.2508$ & $16.7925$ & $16.6886$ \\
\cmidrule{1-7}
\multicolumn{1}{c}{$\TransitionMatrix{2}$}& $5.2853\text{e}^{-3}$ & $5.3018\text{e}^{-3}$ & $5.2953\text{e}^{-3}$ & $13.2122$ & $13.2535$ & $13.2373$ \\
\cmidrule{1-7}
\multicolumn{1}{c}{$\TransitionMatrix{3}$}& $7.5863\text{e}^{-3}$ & $8.1038\text{e}^{-3}$ & $8.0839\text{e}^{-3}$ & $18.9620$ & $20.2582$ & $20.2065$ \\
\bottomrule
\end{tabular}
\caption{Mean costs of the three compared controllers for $\xsys{0} =\lbrack 0\hspace{2mm}  0\rbrack\tr$.}
\label{tab:Costs_1}
\end{table}

\begin{table}[h]
\centering
\begin{tabular}{ccccccc}
\toprule
& \multicolumn{3}{c}{$\wCov_1$} & \multicolumn{3}{c}{$\wCov_2$}\\ \cmidrule{2-7}
& Chizeck et al. & proposed & Vargas et al.& Chizeck et al. & proposed & Vargas et al.\\\cmidrule{1-7}
\multicolumn{1}{c}{$\TransitionMatrix{1}$}& $67.5311$ & $77.9765$ & $72.6012$ & $82.7204$ & $94.7306$ & $89.4974$ \\
\cmidrule{1-7}
\multicolumn{1}{c}{$\TransitionMatrix{2}$}& $65.9251$ & $66.8358$ & $67.5779$ & $79.1335$ & $80.0799$ & $80.8079$ \\
\cmidrule{1-7}
\multicolumn{1}{c}{$\TransitionMatrix{3}$}& $62.6784$ & $70.3547$ & $64.7367$ & $81.6350$ & $90.6844$ & $84.9624$ \\
\bottomrule
\end{tabular}
\caption{Mean costs of the three compared controllers for $\xsys{0} =\lbrack 3\hspace{2mm}  2\rbrack\tr$.}
\label{tab:Costs_2}
\end{table}

An implementation of the presented control law is available at the CloudRunner homepage~\cite{Cloudrunner}.
   	
   \section{Conclusion}
   	\label{sec:Conclusion}
In this paper, we presented a method to compute a constant linear policy for infinite-horizon optimal control of stochastic MJLS with state feedback but without mode observation. To this end, we have rewritten the MJLS dynamics in terms of the second moment, constructed the Hamiltonian, and proposed an iterative algorithm that minimizes the cost function.\\

In the provided numerical example, the proposed control law has only slightly worse performance than the control laws from~\cite{Vargas_2013} and~\cite{Chizeck_1986} although it is mode-independent, time-invariant, and can be precomputed offline.\\

Future work will be concerned with derivation of convergence guarantees for the iterative algorithm and an extension of the proposed approach to measurement-feedback control. Furthermore, an assumption of a more complicated policy structures such as polynomials constitutes another possible research direction.
   	
   \section{ACKNOWLEDGMENTS}
   This work was supported by the German Science Foundation
   (DFG) within the Research Training Group RTG 1194 ``Self-organizing Sensor-Actuator-Networks''.
   	
   \appendix
   \section{Proof of Theorem~\ref{theor:Main}}
   	\label{app:ProofMainTheorem}
 If the dynamics~(\ref{eq:SystemDynamics}) is mean square stabilizable then the second-moment state converges to a fixed point $\Xsysest{i}{\infty} = \lim\limits_{k\rightarrow\infty} \Xsysest{i}{k}$ that is the unique solution of
\begin{align*}
\Xsysest{j}{\infty}
	&=
	\sum\limits_{i=1}^{\NumModes}\transitionprob{ij}\bklammer{(\Asys{i}+\Bsys{i}\ControlLaw)\Xsysest{i}{\infty}(\Asys{i}+\Bsys{i}\ControlLaw)\tr + \modeest{i}{\infty}\Hsys{i}\wCov\Hsys{i}\tr}\ .
\end{align*}
This claim is proven in~\ref{app:ProofConvergence}.\\

Thus, the costs~(\ref{eq:CostFunc}) are finite and can be rewritten as~
\begin{align*}
\CostFunc
	&=
	\sum\limits_{i=1}^{\NumModes} \trace{(\Qsys{i}+\ControlLaw\tr\Rsys{i}\ControlLaw)\Xsysest{i}{\infty}}\ ,
\end{align*}
and the optimization problem~(\ref{eq:OptimizationProblem}) becomes
\begin{align}
\begin{aligned}
\min\limits_{\ControlLaw,\Xsysest{1}{\infty},\dots,\Xsysest{\NumModes}{\infty}}&\hspace{2mm} \trace{(\Qsys{i}+\ControlLaw\tr\Rsys{i}\ControlLaw)\Xsysest{i}{\infty}}\\
\st&\hspace{2mm}
	\Xsysest{j}{\infty}
		=
		\sum\limits_{i=1}^{\NumModes}\transitionprob{ij} \bklammer{(\Asys{i}+\Bsys{i}\ControlLaw)\Xsysest{i}{\infty} (\Asys{i}+\Bsys{i}\ControlLaw)\tr + \modeest{i}{k}\Hsys{i}\wCov\Hsys{i}\tr}\ .
\end{aligned}
\label{eq:OptimizationProblemRestated}
\end{align}
Defining the positive semidefinite Lagrange multiplier $\Lagrange{i}{\infty}$, we obtain the Hamiltonian~$\Hamiltonian$ of~(\ref{eq:OptimizationProblemRestated}) with
\begin{align*}
\Hamiltonian
	&=
	\sum\limits_{i=1}^{\NumModes} \trace{(\Qsys{i}+\ControlLaw\tr\Rsys{i}\ControlLaw)\Xsysest{i}{\infty} + \sum\limits_{j=1}^{\NumModes}\transitionprob{ij}\Lagrange{j}{\infty}\bklammer{(\Asys{i}+\Bsys{i}\ControlLaw)\Xsysest{i}{\infty}(\Asys{i}+\Bsys{i}\ControlLaw)\tr + \modeest{i}{k}\Hsys{i}\wCov\Hsys{i}\tr}\right.\\
	&\left. \vphantom{\sum\limits_{j=1}^{\NumModes}}
	\hspace{2cm}- \Lagrange{i}{\infty}\Xsysest{i}{\infty}}\ .
\end{align*}
Differentiation with respect to $\Xsysest{i}{\infty}$, $\Lagrange{i}{\infty}$, and $\ControlLaw$ yields~(\ref{eq:CoupledEquationsP})-(\ref{eq:CoupledEquationsL}).
%
%
%

   \section{Proof of Second Moment Convergence}
   	\label{app:ProofConvergence}
According to Banach's fixed point theorem, $\{\Xsysest{1}{k}, \dots, \Xsysest{\NumModes}{k}\}$ converges to the unique solution $\{\Xsysest{1}{\infty}, \dots,\Xsysest{\NumModes}{\infty}\}$ for $k\rightarrow\infty$ if (\ref{eq:SecondMomentDynamics}) is a contraction mapping. To show that (\ref{eq:SecondMomentDynamics}) is indeed a contraction mapping if (\ref{eq:SystemDynamics}) is MS-stabilizable, we define the vectorized second moment state vector
\begin{align*}
\vectorizedState{k}
	&=
	\begin{bmatrix}
		\vec{\Xsysest{1}{k}}\tr & \dots & \vec{\Xsysest{\NumModes}{k}}\tr
	\end{bmatrix}\tr\ ,
\end{align*}
where $\vec{\cdot}$ denotes the vectorization operator. The dynamics of $\vectorizedState{k}$ can be written as
\begin{align}
\vectorizedState{k+1}
	&=
	\bigM\vectorizedState{k} + \bigN{k}
\label{eq:ContractionMapping}
\end{align}
with $\bigM$ as in Theorem~\ref{theor:MSS} and
\begin{align*}
\bigN{k}
	&=
	\klammer{\kron{\TransitionMatrix{}\tr}{\identitymatrix}} \diag\begin{bmatrix}\modeest{1}{k}(\kron{\Bsys{1}}{\Bsys{1}}) & \dots & \modeest{\NumModes}{k}(\kron{\Bsys{\NumModes}}{\Bsys{\NumModes}})\end{bmatrix}\\
	&\times
	\begin{bmatrix}\vec{\wCov}\tr & \dots & \vec{\wCov}\tr\end{bmatrix}\tr\ .
\end{align*}
We need to show that
\begin{align*}
\|\vectorizedState{k+1} - \nvec{\phi}_{k+1}\| \leq \beta \|\vectorizedState{k} - \nvec{\phi}_{k}\|\ ,
\end{align*}
where $\beta\in(0;1)$ and
\begin{align*}
\nvec{\phi}_k
	&=
	\begin{bmatrix}
		\vec{\mat{Y}^{(1)}_{k}}\tr & \dots & \vec{\mat{Y}^{(\NumModes)}_{k}}\tr
	\end{bmatrix}\tr
\end{align*}
for any positive semidefinite $\{\mat{Y}^{(1)}_{k},\dots,\mat{Y}^{(\NumModes)}_{k}\}$. Using Lemma~5 from~\cite{Wang_86}, it holds
\begin{align*}
\|\vectorizedState{k+1} - \nvec{\phi}_{k+1}\|
	&=
	\|\bigM\vectorizedState{k} + \bigN{k} -\bigM\nvec{\phi}_{k} - \bigN{k}\|
	=
	\|\bigM(\vectorizedState{k}-\nvec{\phi}_{k})\| = \trace{\bigM\tr\bigM (\vectorizedState{k}-\nvec{\phi}_{k})(\vectorizedState{k}-\nvec{\phi}_{k})\tr}\\
	&\leq 
	\lambda_{\max} \|\vectorizedState{k} - \nvec{\phi}_{k}\|\ ,
\end{align*}
where $\lambda_{\max}$ is the largest eigenvalue of $\bigM\tr\bigM$. Because for $\lambda_{\max} < 1$ (\ref{eq:ContractionMapping}) is a contraction mapping, it has a unique fixed point.\\

Please note that the obtained result corresponds to the stability condition in Theorem~\ref{theor:MSS}.b because $\lambda_{max} = \spectralradius{\bigM}^2$ holds.
   	
   \section{Stabilizability Test}
   	\label{app:StabilizabilityTest}
In order to determine whether the MJLS~(\ref{eq:SystemDynamics}) is MS-stabilizable, we can solve the following optimization problem
\begin{align}
\min\limits_\ControlLaw\ &\ \spectralradius{\mat{M}}\ ,
\label{eq:SpectralRadiusMinimization}
\end{align}
where
\begin{align*}
\mat{M}
	&=
	\klammer{\kron{\TransitionMatrix{}\tr}{\identitymatrix}} 
	\diag\begin{bmatrix}
		\kron{\Aaug{1}}{\Aaug{1}} & \kron{\Aaug{2}}{\Aaug{2}} & \cdots & \kron{\Aaug{\NumModes}}{\Aaug{\NumModes}}
	\end{bmatrix}
\end{align*}
with
\begin{align*}
\Aaug{i}
	&=
	\Asys{i} + \Bsys{i}\ControlLaw\ .
\end{align*}
If the solution $\spectralradiusopt{\mat{M}}$ of~(\ref{eq:SpectralRadiusMinimization}) it holds $\spectralradiusopt{\mat{M}}<1$ then system~(\ref{eq:SystemDynamics}) is MS-stabilizable and we can compute the optimal linear control law according to the numerical algorithm provided in Sec.~\ref{sec:MainResult}. Fig.~\ref{fig:SpectralRadii} illustrates the spectral radii for the system from Sec.~\ref{sec:NumericalExample}. It can be seen that the value function in~(\ref{eq:SpectralRadiusMinimization}) is convex in this scenario.\\

\begin{figure}
\centering
\includegraphics[width=.9\textwidth]{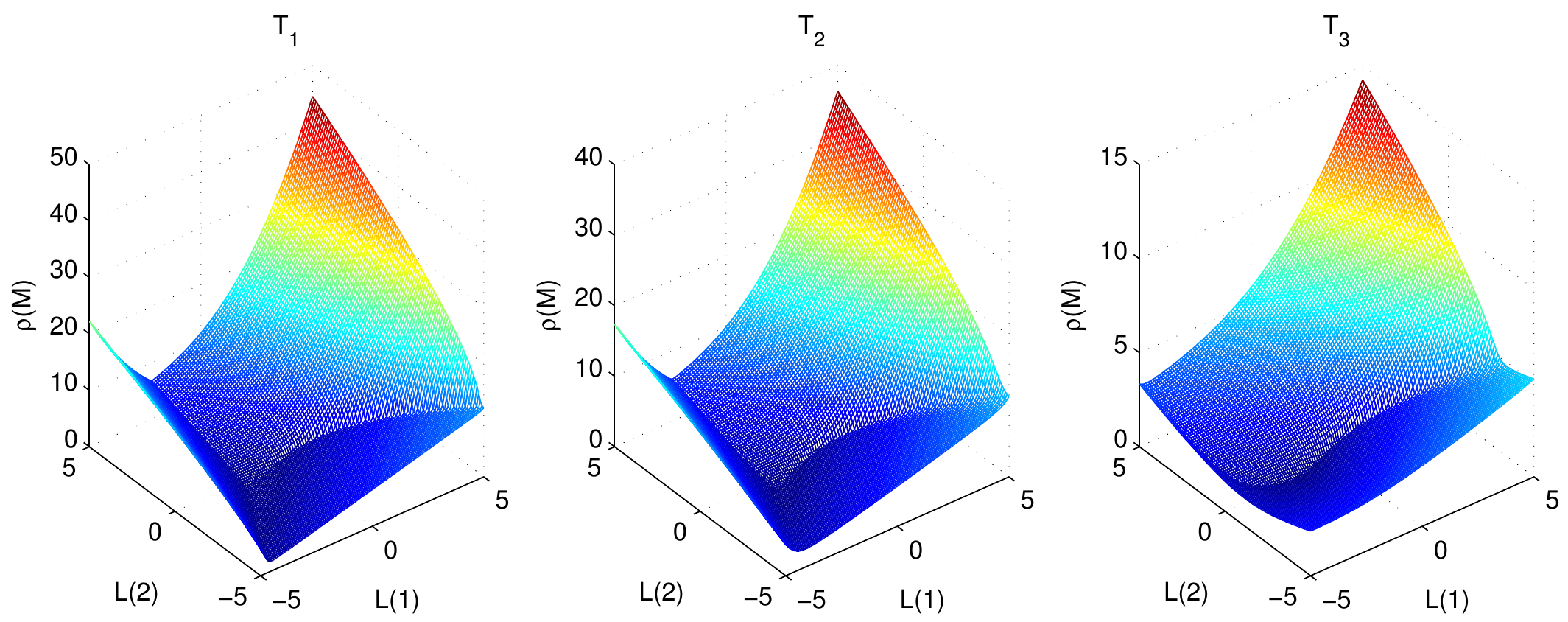}
\caption{Spectral radii for the MJLS from Sec.~\ref{sec:NumericalExample}.}
\label{fig:SpectralRadii}
\end{figure}

However, the value function $\spectralradius{\mat{M}}$ in~(\ref{eq:SpectralRadiusMinimization}) is non-smooth. Thus, we propose to use the smooth convex approximation presented in~\cite{Chen_2004}. The approximation replaces the spectral radius operator $\spectralradius{\mat{A}}$ by
\begin{align}
\mu_\epsilon(\mat{A})
	&=
	\epsilon \sum\limits_{i=1}^{N} \exp\klammer{\frac{\lambda_i}{\epsilon}}\ ,
\label{eq:Approximation}
\end{align}
where $\epsilon>0$ and $\lambda_i$ are the $N$ eigenvalues of $\mat{A}$.\\

Using this approximation, we let $\epsilon$ go from $1$ to $0$ and solve a sequence of optimization problems
\begin{align*}
\min\limits_{\ControlLaw}\ &\ \mu_\epsilon(\bigM)\ .
\end{align*}
Because for the approximation (\ref{eq:Approximation}) it holds
\begin{align*}
\lim\limits_{\epsilon\rightarrow\infty} \mu_\epsilon(\mat{A}) = \spectralradius{\mat{A}}\ ,
\end{align*}
we recover the initial optimization problem (\ref{eq:SpectralRadiusMinimization}) as $\epsilon$ goes to zero. Additionally, we can use the gradient and the Hessian given in~\cite{Chen_2004}.
    
    \bibliographystyle{IEEEtran}
    \bibliography{Sections/00_Literature}
\end{document}